\documentclass[smallextended]{svjour3}
\smartqed
\usepackage{amsmath,amssymb,paralist,subfigure,graphicx,cite}
\usepackage{algorithmic}
\usepackage{color}

\usepackage{stmaryrd,mathrsfs,url}
\usepackage{cuted,mathtools,lipsum}
\newtheorem{lem}{Lemma}
\newtheorem{ass}{Assumption}
\newtheorem{thm}{Theorem}

\newtheorem{rem}{Remark}

\newtheorem{cor}{Corollary}

\def\mb{\mathbf}
\def\mbb{\mathbb}
\def\mc{\mathcal}

\DeclareMathOperator*{\argmin}{argmin}

\begin{document}
	
	\title{ D-SVM over Networked Systems with Non-Ideal Linking Conditions 
	}
	\author{Mohammadreza Doostmohammadian        \and
	Alireza~Aghasi \and
	Houman~Zarrabi
	}
	
	\institute{ M.  Doostmohammadian \at
		Faculty of Mechanical Engineering, Semnan University, Semnan, Iran. \\
		Tel.: +98-23-31533429\\
		\email{doost@semnan.ac.ir}
		\and
		A.  Aghasi \at Electrical Engineering and Computer Science Department, Oregon  State University, USA. \\ 
		\email{alireza.aghasi@oregonstate.edu}
		\and
		H. Zarrabi \at Iran Telecom research Center (ITRC), Tehran, Iran.  \\ 
		\email{h.zarrabi@itrc.ac.ir}
	}
	
	\date{Received: date / Accepted: date}

	\maketitle
	
	\begin{abstract}
		This paper considers distributed optimization algorithms, with application in binary classification via distributed support-vector-machines (D-SVM) over multi-agent networks subject to some link nonlinearities. The agents solve a consensus-constraint distributed optimization cooperatively via continuous-time dynamics, while the links are subject to strongly sign-preserving odd nonlinear conditions. Logarithmic quantization and clipping (saturation) are two examples of such nonlinearities.  In contrast to existing literature that mostly considers ideal links and perfect information exchange over linear channels, we show how general sector-bounded models affect the convergence to the optimizer (i.e., the SVM classifier) over dynamic balanced directed networks. In general, any odd sector-bounded nonlinear mapping can be applied to our dynamics. The main challenge is to show that the proposed system dynamics always have one zero eigenvalue (associated with the consensus) and the other eigenvalues all have negative real parts. This is done by recalling arguments from \textit{matrix perturbation theory}. Then, the solution is shown to converge to the agreement state under certain conditions. For example, the gradient tracking (GT) step size is tighter than the linear case by factors related to the upper/lower sector bounds. To the best of our knowledge, no existing work in distributed optimization and learning literature considers non-ideal link conditions.      
		
	\end{abstract}
\maketitle

\section{Introduction}\label{sec_intro}
\textbf{Background:}
Support-vector-machine (SVM) is a supervised learning method used for classification, regression, and outlier detection \cite{cortes1995support,drucker1996support}. Consider a labelled set of given data points each belonging to one of the two classes/labels. The SVM classifies these data and decides which class a new data point belongs to. In the centralized setting, all the data points are available at a central processing node/unit which performs the classification task and finds the SVM parameters. In distributed/decentralized setting (referred to as D-SVM), the data points are distributed among some computing nodes. Every computing node has access to a portion of the data points and computes the associated SVM parameters for classification. Then shares the SVM information over the multi-agent network (that could be subject to nonlinearities) and updates the SVM parameters based on the received information over the network. In this way, the nodes compute the SVM classifiers collaboratively without sharing any data points. The existing works of literature are either (i) centralized and need all the data points at a central processor, or (ii) need the missing data points at all the computing nodes. Both cases require a lot of information sharing and long-distance communication to the central unit (case (i)) or other computing nodes (case (ii)). Also, these may raise privacy concerns due to sharing raw data over the network. The distributed setups (D-SVM), on the other hand, only share SVM data to be updated. Many existing D-SVM solutions consider a linear setup, where the information exchange among the nodes is ideal and linear. However, this data-sharing might be subject to non-ideal linking conditions (e.g., quantization or clipping).  This work takes these into consideration by modelling nonlinear linking among the nodes. In general, we assume the data received from one node from another node has gone through a nonlinear channel and this nonlinearity represents any non-ideal linking condition that may exist in the real-world. This gives the most general form of non-ideal linking and to our best knowledge, no existing work in the literature considers such a general solution. This motivates our paper for practical D-SVM and machine learning applications. 

\textbf{Literature Review:} This paper extends our previous results on consensus optimization and distributed learning in \cite{dsvm} to address some practical nonlinearities in real applications. For example, in many machine-learning applications, the data exchange over the network is in discrete-values \cite{ecc_mixed} and needs to be quantized \cite{magnusson2018communication}, or in swarm robotic networks the actuators might be subject to saturation \cite{ccta22}. These inherent unseen nonlinearities in the real optimization models may result in inaccuracy and unacceptable outcomes, degrading efficiency and increasing the operation costs \cite{ccta22}. On the other hand, in some applications, e.g., optimal resource allocation and consensus algorithms, sign-based nonlinearities are added for the purpose of robustness to impulsive noise \cite{stankovic2020nonlinear} or improving the convergence rate \cite{garg2019fixed2,my_ecc,ning2017distributed,rahili_ren}. 
Many other recent works address different constraints in the distributed optimization framework. For example, in contrast to primary linear models \cite{van2019distributed,nedic2017achieving,simonetto2017decentralized,akbari2015distributed,ling2013decentralized,wei_me_cdc22,koppel2018decentralized,xi2017add,forero2010consensus},  consensus optimization in a multiple leader containment as in opinion dynamics \cite{ecc_containment}, under communication constraint and server failure \cite{ecc_compres}, under mixed-integer value algorithms over computer networks \cite{ecc_mixed}, and under non-negative constraints on the states (non-negative orthant) \cite{ecc_nonnegative} are discussed recently.

\textbf{Contributions:} In this work, we assume that the links (over the multi-agent network) are subject to sector-bounded link nonlinear conditions, i.e., the sent information over the transmission links is not linearly delivered; for example, due to (log) quantization or clipping. We show that under strongly sign-preserving odd nonlinearity on the links, the exact convergence to the optimizer can be achieved. Some detailed examples of such nonlinear applications are given in our previous work on resource allocation \cite{ccta22}. Similar to \cite{van2019distributed,nedic2017achieving,simonetto2017decentralized,akbari2015distributed} we prove convergence over (possibly) dynamic networks via a hybrid setup based on matrix perturbation theory and algebraic graph theory. For general strictly convex cost functions, we prove convergence of the nonlinear version of our previously proposed protocol in \cite{dsvm} over general \textit{weight-balanced} undirected graphs, which advances the \textit{doubly-stochastic} assumption in some ADMM-based solutions \cite{ling2013decentralized,wei_me_cdc22}. Note that, in general, weight symmetry (and weight balancing) is easier to satisfy in contrast to weight-stochasticity. This is particularly of interest in packet-drop and link-failure scenarios \cite{icrom22_drop} which require weight-compensation algorithms for convergence \cite{6426252,cons_drop_siam} in the existing bi-stochastic network designs. In other words, the existing bi-stochastic networks need compensation algorithms as described in \cite{6426252,cons_drop_siam} to redesign the network weights after link removal, but the weight-symmetric condition in this work may still hold with no need for redesign algorithms, e.g., in the case of bidirectional link removal. This is discussed in detail for the distributed resource allocation problem in \cite{icrom22_drop}.

Moreover, we show that the gradient tracking property of our proposed GT-based dynamics holds under additive odd sign-preserving nonlinearities, which is another improvement over ADMM-based solutions. This further allows for, e.g., adding sign-based dynamics to reach faster convergence (in finite/fixed-time \cite{taes2020finite}) or to make the solution resilient to outliers, disturbances, and impulsive noise as in some consensus literature \cite{stankovic2020nonlinear}, but with the drawback of unwanted chattering due to non-Lipschitz continuity at the origin. To the best of our knowledge, such general resiliency to additive linking conditions is not addressed in the existing distributed optimization and learning literature. 

\textbf{Paper organization:} Section~\ref{sec_pre} recaps some preliminaries on algebraic graph theory and perturbation theory. Section~\ref{sec_prob} formulates the D-SVM as a consensus optimization problem. Section~\ref{sec_dyn} presents our nonlinear GT-based method with its convergence discussed in Section~\ref{sec_conv}. Section~\ref{sec_sim} gives a simulation on an illustrative example, and  Section~\ref{sec_con} concludes the paper.   

\section{Preliminaries} \label{sec_pre}
\subsection{Notations}
${\lVert A\rVert}$ denotes the norm operator on a matrix, i.e.,~$\lVert A\rVert= \sup_{\lVert \mb{x} \rVert =1}  \lVert A \mb{x} \rVert$. ${\lVert A\rVert_2}$ is the spectral norm defined as the square root of the max eigenvalue of $A^\top A$, and $\lVert A\rVert_{\infty} = \max_{1\leq i\leq n} \sum_{j=1}^n |a_{ij}|$ as the infinity norm. $\sigma(A)$ denotes the eigen-spectrum of $A$ and $\lambda$ denotes the eigenvalue. LHP and RHP respectively imply left-half-plane and right-half-plane in the complex domain.
$\partial_a z = \frac{dz}{da}$ denotes the derivative with respect to $a$.
$\mb{1}_n$ and~$\mb{0}_n$ are size~$n$ vectors of all~$1$'s and~$0$'s. `$;$' denotes column vector concatenation. $\nabla_{\mb{x}} F$ and $H:=\nabla_{\mb{x}}^2 F$ denote the first and second gradient of $F$ with respect to $\mb{x}$. Operators $\prec,\preceq,\succ,\succeq$ denote the element-wise version of $<,\leq,>,\geq$. The identity matrix of size $n$ is denoted by $I_n$.

\subsection{Recall on Algebraic Graph Theory}
We represent the multi-agent network by a strongly-connected undirected\footnote{In the case of linear and ideal links, the solution holds for strongly connected directed graphs (digraphs), see \cite{dsvm}.} graph~$\mc{G}$ with adjacency matrix $W=\{w_{ij}\}$. $w_{ij}$ is defined as the weight on the link $j\rightarrow i$ and zero otherwise. For an SC $\mc G$, matrix $W$ is irreducible. The Laplacian matrix~$\overline{W}=\{\overline{w}_{ij}\}$ is then defined as $\overline{w}_{ij}=w_{ij}$ for $i\neq j$ and $\overline{w}_{ij}=-\sum_{i=1}^n w_{ij}$ for $i=j$.
\begin{lem} \label{lem_sc}
	\cite{SensNets:Olfati04} The Laplacian $\overline{W}$ of a connected weigh-balanced graph
	has one isolated eigenvalue at zero and the rest on LHP.
\end{lem}

A graph is weight-balanced (WB) if the weight-sum of incoming and outgoing links at every node $i$ are equal, i.e.,~$\sum_{j=1}^n w_{ji} = \sum_{j=1}^n w_{ij}$. 
\begin{lem} \label{lem_laplacian}
	\cite{SensNets:Olfati04} Given a WB-graph, $W \mb{1}_n = \mb{1}_n^\top W$, and its Laplacian~$\overline{W}$ has the same left and right eigenvectors $\mb{1}_n^\top$ and~$\mb{1}_n$ associated with its zero eigenvalue, i.e.,~$\mb{1}_n^\top \overline{W}= \mb{0}_n$ and~$\overline{W}\mb{1}_n=\mb{0}_n$.
\end{lem}

\subsection{Auxiliary Results on Perturbation Theory} \label{sec_aux}
We recall some results of the matrix perturbation theory.
\begin{lem} \label{lem_dM} {\hspace{-0.005cm}~\cite{stewart_book,cai2012average}}
Assume an~$n$-by-$n$ matrix~$P(\alpha)$ which smoothly depends on a real parameter~${\alpha \geq 0}$, and has~${l<n}$ equal eigenvalues~$\lambda_1=\ldots=\lambda_l$, associated with linearly independent right and left unit eigenvectors~$\mb{v}_1,\ldots,\mb{v}_l$ and~$\mb{u}_1,\ldots,\mb{u}_l$. 
Denote by~$\lambda_i(\alpha)$, corresponding to~${\lambda_i, i \in \{1,\ldots,l\}}$, as the eigenvalues of~$P(\alpha)$ as a function of independent parameter~$\alpha$. Then,~$\partial_{\alpha} \lambda_i|_{\alpha=0}$ is the $i$-th eigenvalue of the following~$l$-by-$l$ matrix, and
\[\left(\begin{array}{ccc}
\mb{u}_1^\top P' \mb{v}_1 & \ldots & \mb{u}_1^\top P' \mb{v}_l \\
 & \ddots & \\
\mb{u}_l^\top P' \mb{v}_1 & \ldots & \mb{u}_l^\top P' \mb{v}_l 
\end{array} \right), ~ P' = \partial_{\alpha} P(\alpha)|_{\alpha=0}
\]  
\end{lem}

\begin{lem} \label{lem_dbound} \cite[Theorem~39.1]{bhatia2007perturbation}
Let system matrix~${M(\alpha) = M_0 +\alpha M_1}$, with $\alpha$-independent matrices $M_0$ and $M_1$. Then, the optimal matching distance satisfies ${d(\sigma(M),\sigma(M_0))\leq 4(\lVert M_0\rVert+\lVert M\rVert)^{1-\frac{1}{nm}} \lVert \alpha M_1\rVert^{\frac{1}{nm}}}$ for~${\min_{\pi} \max_{1\leq i\leq 2nm} (\lambda_i - \lambda_{\pi(i)}(\alpha))}$ with~$\pi(i)$ as the~$i$th permutation over $2mn$ symbols.
\end{lem}

\section{Problem Statement} \label{sec_prob}
\subsection{Mathmatical Formulation} \label{sec_formul}
Following the notation in \cite{dsvm}, let~$\mb{x}= [\mb{x}_1;\mb{x}_2;\dots;\mb{x}_n] \in \mathbb{R}^{nm}$ represent the column state vector  and $\mb{x}_i \in \mathbb{R}^m$ represent the state at agent $i$. The following consensus-constrained distributed optimization problem is considered
\begin{align}\nonumber
\min_{\mb{x} \in \mathbb{R}^{nm}} &
 F(\mb x) = \sum_{i=1}^{n} f_i(\mb{x}_i)\\\label{eq_prob} 
 \text{subject to}& ~ \mb{x}_1 = \mb{x}_2 = \dots = \mb{x}_n.
\end{align}
with $f_i:\mathbb R^{m} \mapsto \mathbb R$ as the local cost function at agent/node $i$. Let $\nabla F:\mathbb R^{nm} \mapsto  \mathbb R^{nm}$ denote the gradient of global cost~${F:\mbb R^{nm} \mapsto  \mbb R}$. The idea in this work is to solve problem~\eqref{eq_prob} in a distributed way, i.e., locally over a multi-agent network. In this scenario, every agent only uses the local information in its neighbourhood (and its own information) to solve \eqref{eq_prob}.We assume that problem \eqref{eq_prob} is proper and well-defined and  there exists $\overline{\mb{x}}^* \in \mathbb{R}$ such that~${\nabla F(\mb{x}^*) = \mb{0}_{n}}$. 
Then, the optimizer of \eqref{eq_prob} is defined as~$\mb x^* := [\mb{x}^*_1;\mb{x}^*_2;\dots;\mb{x}^*_n] = \overline{\mb{x}}^* \mb{1}_n$ which satisfies ${\sum_{i=1}^{n} \partial_{x_i} f_i(\mb{x}^*_i) = 0}$. 

\begin{rem}
Other equivalent formulations of the consensus-constraint problem \eqref{eq_prob} are also discussed in the literature. In \cite{gharesifard2013distributed}, problem \eqref{eq_prob} is shown to be equivalent with general \textit{unconstrained} formulation,
\begin{align}\label{eq_prob_general} 
\min_{\mb{x} \in \mathbb{R}^{m}} &
 F(\mb x) = \sum_{i=1}^{n} f_i(\mb{x})
\end{align}
or its Laplacian-constraint version as,
\begin{align}\nonumber
\min_{\mb{x} \in \mathbb{R}^{nm}} &
 F(\mb x) = \sum_{i=1}^{n} f_i(\mb{x}_i)\\\label{eq_prob_laplac} 
 \text{subject to}& ~\mc{L}\mb{x} = \mb{0}_{nm}.
\end{align}
with $\mc{L} = L \otimes I_m$ and $L$ as the Laplacian of the weight-balanced graph $\mc{G}$. 
\end{rem}

\begin{ass} \label{ass_f(x)}
Each local cost~$f_i$ is strictly convex with smooth gradient and positive-definite (PD) $\nabla^2 f_i(\mb{x}_i)$. 
\end{ass} 

Under Assumption~\ref{ass_f(x)}, the unique optimizer of \eqref{eq_prob} is in the form $\mb{x}^* = \mb{1}_n \otimes \overline{\mb{x}}^*$ with $\overline{\mb{x}}^* \in \mathbb{R}^m$. 

\subsection{Application in D-SVM} \label{sec_dsvm}
We recall here a specific application of problem~\eqref{eq_prob} given in~\cite{dsvm} for distributed support-vector-machines (D-SVM). 
Consider binary classification using a given training set of data points~${\boldsymbol{\chi}_i \in \mathbb{R}^{m-1}}$, ${i=1,\ldots,N}$, each labelled by~${l_i \in \{-1,1\}}$. The standard (centralized) SVM problem finds the optimal max-margin classifier ${\boldsymbol{\omega}^\top \boldsymbol{\chi} - \nu =0}$ (as a hyperplane) to partition the data points into two classes in $\mathbb{R}^{m-1}$. The hyperplane parameters $\boldsymbol{\omega}~$ and~$\nu$ are optimal solution to the following minimization problem \cite{chapelle2007training}: 
\begin{align} \label{eq_svm_cent}
	\begin{aligned}
		\displaystyle
		& \min_{\boldsymbol{\omega},\nu}
		~ &  \boldsymbol{\omega}^\top \boldsymbol{\omega} + C \sum_{j=1}^{N} \max\{1-l_j( \boldsymbol{\omega}^\top \phi(\boldsymbol{\chi}_j)-\nu),0\}^p
	\end{aligned}
\end{align}
with~${p = \{1,2,\ldots\}}$ and $C$ respectively as the smoothness  and   margin-size parameters.
Given the parameters $\boldsymbol{\omega},\nu$ of the separating hyperplane, a test data point~$\widehat{\boldsymbol{\chi}}$ is labeled (classified) by~${g(\widehat{\boldsymbol{\chi}})= \text{sgn}(\boldsymbol{\omega}^\top \widehat{\boldsymbol{\chi}} - \nu)}$. The notions of nonlinear mapping~$\phi(\cdot)$
associated with a \textit{kernel function}~$K(\boldsymbol{\chi}_i,\boldsymbol{\chi}_j)=\phi(\boldsymbol{\chi}_i)^\top \phi(\boldsymbol{\chi}_j)$ can further be used to project the data points into a high-dimensional space~$\mc{F}$ in which the data points are linearly separable such that~${g(\widehat{\boldsymbol{\chi}})= \text{sgn}(\boldsymbol{\omega}^\top \phi(\widehat{\boldsymbol{\chi}}) - \nu)}$ determines the class of~$\widehat{\boldsymbol{\chi}}$. 
Recall that, as a standard convention, the non-differentiable function $\max\{z,0\}^p$ for~${p=1}$ in SVM formulation~\eqref{eq_svm_cent} is replaced by the twice differentiable~${L(z,\mu)=\frac{1}{\mu}\log (1+\exp(\mu z))}$ (hing loss model \cite{garg2019fixed2,slp_book}).  
It can be shown that by setting~$\mu$ large enough~${L(z,\mu)}$ becomes arbitrarily close to~$\max\{z,0\}$~\cite{slp_book,dsvm}.

In D-SVM setup, the data points are distributed among~$n$ agents each with~$N_i$ data points denoted  by~$\boldsymbol{\chi}^i_j,{j=1,\ldots,N_i}$. The trivial solution is that agents share some data points (the support vectors) locally which raises data-privacy concerns \cite{navia2006distributed}. The more recent idea is to find a distributed learning mechanism such that no agent reveals its local (support vector) data to any other agent or over non-private communication channels \cite{dsvm}. Each agent $i$ solves \eqref{eq_svm_cent} and finds its local classifier parameters $\boldsymbol{\omega}_i$ and~$\nu_i$ based on its partial dataset~$\boldsymbol{\chi}^i_j,{j=1,\ldots,N_i}$, which in general
differs from other agents. 
This is done by formulating the problem into the general distributed optimization format as given by \eqref{eq_prob}, 
\begin{align} \label{eq_svm_dist}
	\begin{aligned}
		\displaystyle
		 \min_{\boldsymbol{\omega}_1,\nu_1,\ldots,\boldsymbol{\omega}_n,\nu_n}
		\quad &  \sum_{i=1}^{n} \boldsymbol{\omega}_i^\top \boldsymbol{\omega}_i + C \sum_{j=1}^{N_i} \tfrac{1}{\mu}\log (1+\exp(\mu z)) \\
		 \text{subject to} \quad&  \boldsymbol{\omega}_1 = \dots = \boldsymbol{\omega}_n,\qquad\nu_1 = \dots =\nu_n, 
	\end{aligned}
\end{align} 
with hing loss function~${f_i:\mathbb{R}^m \mapsto \mbb R}$ locally defined at each agent $i$. The consensus constraint in \eqref{eq_svm_dist} ensures that agents learn (and reach an agreement on) the overall classifier parameters $\boldsymbol{\omega}$ and~$\nu$.

\section{Proposed Dynamics subject to Link Nonlinearity}\label{sec_dyn}
\subsection{The Linear Setup}
To solve problem~\eqref{eq_prob}, we previously proposed the following GT-based linear dynamics in \cite{dsvm}.
\begin{align} \label{eq_xdot}	 
	\dot{\mb{x}}_i &= -\sum_{j=1}^{n} w_{ij}^q(\mb{x}_i-\mb{x}_j)-\alpha \mb{y}_i, \\ \label{eq_ydot}
	\dot{\mb{y}}_i &= -\sum_{j=1}^{n} a_{ij}^q(\mb{y}_i-\mb{y}_j) + \partial_t  \nabla f_i(\mb{x}_i),
\end{align}
with~$\mb{x}_i(t)$ as the state of agent~$i$ at time~${t\geq 0}$, ${\mb{y}=[\mb{y}_1;\mb{y}_2;\dots;\mb{y}_n]\in\mbb{R}^{mn}}$ as the auxiliary variable tracking the sum of local gradients, and the matrix~${A_q=\{a_{ij}^q\}},{W_q=\{w_{ij}^q\}}$ as the WB adjacency matrices of~$\mc{G}_q$ ($q$ as possible switching signal), and~$\alpha>0$ as the GT step-size. The compact formulation of \eqref{eq_xdot}-\eqref{eq_ydot} is
\begin{align} \label{eq_xydot}
	\left(\begin{array}{c} \dot{\mb{x}} \\ \dot{\mb{y}} \end{array} \right) &= M(t,\alpha,q) \left(\begin{array}{c} {\mb{x}} \\ {\mb{y}} \end{array} \right),\\ \label{eq_M}
M(t,\alpha,q) &= \left(\begin{array}{cc} \overline{W}_q \otimes I_m & -\alpha I_{mn} \\ H(\overline{W}_q\otimes I_m) & \overline{A}_q \otimes I_m - \alpha H
\end{array} \right).
\end{align}
with Hessian matrix~$H:=\text{blockdiag}[\nabla^2 f_i(\mb{x}_i)]$.
Note that the linear networked dynamics~\eqref{eq_xydot}-\eqref{eq_M} and its nonlinear counterpart in the next section represent \textit{hybrid dynamical systems} because the matrix~$H$ (and parameter $\Xi(t)$ in the upcoming nonlinear case~\eqref{eq_xdot_g}-\eqref{eq_beta_M}) varies in continuous-time and the structure of the Laplacians $\overline{W}_q$ and~$\overline{A}_q$ may switch in discrete-time via a switching signal $q$ (as a dynamic network topology). 
We assume that this signal $q$ as the \textit{jump map} satisfies the assumptions for stability and regularity conditions as discussed in \cite{goebel2009hybrid,sam_tac:17}. We skip the details due to space limitations and refer interested readers to our previous work \cite[footnote~1]{dsvm} for more discussions on this. 

\subsection{The Nonlinear Setup}
Next, we consider possible nonlinearities at the linking between agents, denoted by nonlinear mapping $g:\mbb R^m \mapsto \mbb R^m $. Then, \eqref{eq_xdot}-\eqref{eq_ydot} is reformulated as\footnote{The problem can be extended to consider node nonlinearities (non-ideal agents) as $w_{ij}^q g(\mb{x}_i-\mb{x}_j)$ and $a_{ij}^q g(\mb{y}_i-\mb{y}_j)$ over undirected weight-symmetric networks. Examples of distributed resource allocation are given in \cite{ccta22}. }
\begin{align} \label{eq_xdot_g}	 
	\dot{\mb{x}}_i &= -\sum_{j=1}^{n} w_{ij}^q (g(\mb{x}_i)-g(\mb{x}_j))-\alpha \mb{y}_i, \\ \label{eq_ydot_g}
	\dot{\mb{y}}_i &= -\sum_{j=1}^{n} a_{ij}^q (g(\mb{y}_i)-g(\mb{y}_j) ) + \partial_t \nabla f_i(\mb{x}_i),
\end{align}
\begin{ass} \label{ass_g_sector}
	The nonlinear mapping $g(\cdot)$ is odd, sign-preserving, and monotonically non-decreasing. Further, $g$ is upper/lower sector-bounded by positive parameters $\kappa,\mc{K}>0$, i.e., $0<\kappa \leq \frac{g(z)}{z} \leq \mc{K}$, where the parameters $\kappa,\mc{K}$ satisfy the frequency condition in circle criterion\footnote{The circle criterion for absolute stability is as follows: given a nonlinear system in the form $\dot{\mb{x}} = A\mb{x} + \phi(\mb{v},t)$, $\mb{v} = C\mb{x}$ and conditions (i) $ \kappa \mb{v} \leq \phi(\mb{v},t) \leq \mc{K} \mb{v}$, $\forall \mb{v},t$; (ii) $det(i \omega I_n - A) \neq 0$, $\forall \omega \in \mathbb{R}^+ $ and $\exists k_0 \in [\kappa,\mc{K}]$, $A+k_0C$ is stable; (iii) $ \Re[(\mc{K}C(i \omega I_n - A)^{-1}-1)(1-\mc{K}C(i \omega I_n - A)^{-1})] < 0$, $\forall \omega \in \mathbb{R}^+ $ (known as the frequency condition), then $\exists c_1,c_2>0 $ such that for any solution of the system we have $ |\mb{x}(t)| \leq c_1 \exp{(-c_2t)} |\mb{x}(0)|$, $\forall t>0$.}.
\end{ass}
In compact form at \textit{every operating time} $t$, one can rewrite the dynamics  \eqref{eq_xdot_g}-\eqref{eq_ydot_g} as
\begin{align} \label{eq_xydot1}
	\left(\begin{array}{c} \dot{\mb{x}} \\ \dot{\mb{y}} \end{array} \right) = M_g(t,\alpha,q) \left(\begin{array}{c} {\mb{x}} \\ {\mb{y}} \end{array} \right),
\end{align} 
\begin{align} \label{eq_M_g}
M_g(t,\alpha,q) = \left(\begin{array}{cc} \overline{W}_{q,\Xi} \otimes I_m & -\alpha I_{mn} \\ H(\overline{W}_{q,\Xi}\otimes I_m) & \overline{A}_{q,\Xi} \otimes I_m - \alpha H
\end{array} \right).
\end{align}
where $M_g(t,\alpha,q)$ represents the linearization at every operating point $t$, which is also a function of the parameter $\alpha$ and a switching signal $q$ of the network (to be described later).
In the linear case as in \cite{dsvm} with $g(z)=z$, one can formulate the above as  $M = M^0 + \alpha M^1$ with 
\begin{eqnarray}\nonumber
	M^0 &=&   \left(\begin{array}{cc} \overline{W}_q \otimes I_m & \mb{0}_{mn\times mn} \\ H(\overline{W}_q \otimes I_m) & \overline{A}_q\otimes I_m \end{array} \right),\\\nonumber
	M^1 &=& \left(\begin{array}{cc} \mb{0}_{mn\times mn} & - {I_{mn}} \\ {\mb{0}_{mn\times mn}} & - H \end{array} \right),
\end{eqnarray}
Then, in the nonlinear case of this paper, one can rewrite the compact formulation of
\eqref{eq_xydot1}-\eqref{eq_M_g} as,
\begin{align}  \label{eq_Mg}
M_g(t,\alpha,q) &=   M_g^0 + \alpha M^1 \\
\kappa M^0 & \preceq M_g^0 \preceq \mc{K} M^0 \\ \label{eq_beta_M}
M_g^0 = \Xi(t)  M^0 &,~ \kappa I_n  \preceq \Xi(t) \preceq \mc{K} I_n 
\end{align}
with $\Xi(t) = \mbox{diag}[\xi(t) ]$, $\xi(t) = [\xi_1(t);\xi_2(t);\dots;\xi_n(t)]$, and $\xi_i(t) = \frac{g(\mb{x}_i)}{\mb{x}_i}$. Note that, from Assumption~\ref{ass_g_sector}, $\kappa \leq \xi_i(t) \leq \mc{K}$, and we have $\overline{W}_{q,\Xi} =  \overline{W}_{q} \Xi(t)$, $\overline{A}_{q,\Xi} =  \overline{A}_{q} \Xi(t)$ as one can write (with a slight abuse of notation) $g(\mb{x}(t)) = \Xi(t) \mb{x}(t)$. 
\begin{rem} \label{rem_lambda_xi}
    To relate $\sigma(\overline{W}_{q,\Xi})$ and  $\sigma(\overline{W}_{q})$, note that $\mbox{det}(\overline{W}_{q,\Xi}-\lambda I_{2mn}) = \mbox{det}(\overline{W}_{q}-\lambda \Xi(t)^{-1})$ since $\Xi(t)$ is a diagonal matrix. 
\end{rem}
In the rest of this paper, for notation simplicity, we drop the dependence on~$(t,\alpha,q)$ unless where needed.

\begin{ass} \label{ass_W}
		The network is connected and undirected (bidirectional links) with non-negative weights~${W_q=\{w_{ij}^q\}}$ and~${A_q=\{a_{ij}^q\}}$. Further,~$\sum_{j=1}^n w_{ij}^q<1$ and ~${\sum_{j=1}^n a_{ij}^q<1}$ with $q$ denoting the switching signal in hybrid mode.
\end{ass}


Assumptions~\ref{ass_g_sector}-\ref{ass_W} guarantee that under~\eqref{eq_xdot_g} and~\eqref{eq_ydot_g}: 
\begin{align}  \label{eq_sumydot}
	\sum_{i=1}^n \dot{\mb{y}}_i 
 = \sum_{i=1}^n \partial_t \nabla f_i(\mb{x}_i), ~~
	\sum_{i=1}^n \dot{\mb{x}}_i 
	= -\alpha \sum_{i=1}^n\mb{y}_i.
\end{align}
which can be proved similar to the proof of \cite[Lemma~3]{my_ecc}.
Initializing the auxiliary variable~$\mb{y}(0)=\mb{0}_{nm}$ it can be also seen that, similar to the linear case \eqref{eq_xdot}-\eqref{eq_ydot} \cite{dsvm}, we have
\begin{eqnarray} \label{eq_sumxdot2}
	\sum_{i=1}^n \dot{\mb{x}}_i = -\alpha \sum_{i=1}^n\mb{y}_i = -\alpha \sum_{i=1}^n \nabla f_i(\mb{x}_i),
\end{eqnarray}
and the time-derivative of $\sum_{i=1}^n \mb{x}_i$ (the sum of states) moves towards the summed gradient. Note that the same set of equations holds for linear case \eqref{eq_xdot_g}-\eqref{eq_ydot_g}. This follows from the structure of our GT-based protocol that allows for odd sign-preserving nonlinearities to be embedded in the model without disrupting the gradient tracking dynamics. This particularly is an improvement over existing ADMM-based dynamics, e.g., \cite{ling2013decentralized}.
Then, one can show that for any initialization~$\mb{x}(0) \notin \mbox{span} \{\mb{1}_n \otimes \varphi\}$ and~$\mb{y}(0)=\mb{0}_{nm}$, from~\eqref{eq_sumxdot2}, and due to the strict convexity of~$F(\mb{x})$, the following uniquely holds at~${\mb{x}=\mb{x}^*=\mb{1}_n \otimes \overline{ \mb{x}}^*}$,
\[\sum_{i=1}^n \dot{\mb{x}}_i = -\alpha (\mathbf 1_n^\top \otimes I_m) \nabla F(\mb{x}^*) =  \mb{0}_m. \]
Further, from~\eqref{eq_xdot} we have~$\dot{\mb{x}}_i = \mb{0}_m$ and from~\eqref{eq_ydot}, 
\[\dot{\mb{y}}_i = \partial_t \nabla f_i(\overline{ \mb{x}}^*) =  \nabla^2 f_i(\overline{ \mb{x}}^*) \dot{\mb{x}}_i = \mb{0}_m, 	\] 
which shows that the  state~$[\mb{x}^*;\mb{0}_{nm}]$ with $(\mathbf 1_n^\top \otimes I_m) \nabla F(\mb{x}^*) = \mb{0}_m$ is an invariant equilibrium point of both linear dynamics~\eqref{eq_xdot}-\eqref{eq_ydot} and nonlinear dynamics~\eqref{eq_xdot_g}-\eqref{eq_ydot_g} satisfying Assumption~\ref{ass_g_sector}.

\section{Convergence Analysis via Perturbation Theory} \label{sec_conv}
To prove the convergence to the optimizer $\mb{x}^*$ under the proposed GT-based nonlinear (hybrid) setup, we show the stability of the matrix~$M_g$ at every time-instant using the matrix perturbation theory~\cite{stewart_book}. Recall that the stability of a nonlinear system can be analyzed by addressing its stability at every operating point (over time) \cite{nonlin,Liu-nature}. In this direction, we specifically show that the algebraic multiplicity of zero eigenvalues of~$M_g$ is $m$ (at all times $t>0$) while the remaining eigenvalues are in the LHP. For this part, unlike existing literature with \textit{strong convexity} condition
\cite{van2019distributed,nedic2017achieving,ling2013decentralized,simonetto2017decentralized,akbari2015distributed,koppel2018decentralized}, we only assume strict convexity (Assumption~\ref{ass_f(x)}). We also show that the convergence rate of the solution and its associated Lyapunov function depend on the largest non-zero eigenvalue of $M_g$ in the LHP; and, following the continuity of the Lyapunov function at the jump points, the convergence over the entire (hybrid) time horizon can be proved \cite{nonlin,goebel2009hybrid}. 

\begin{thm} \label{thm_zeroeig}
Let Assumptions~\ref{ass_f(x)}-\ref{ass_W} hold. At every iteration point $ \forall t>0$, the matrix $M_g$ given by \eqref{eq_M_g} has only one set of zero eigenvalues with algebraic multiplicity $m$. 
\end{thm}
\begin{proof}
The proof follows similar to \cite[Theorem~1]{dsvm} for linear dynamics~\eqref{eq_xdot}-\eqref{eq_ydot}. For the nonlinear case~\eqref{eq_xdot_g}-\eqref{eq_ydot_g} we have $M_g = M^0_g + \alpha M^1$, and the eigen-spectrum follows as
\begin{align} \label{eq_spect_k}
\kappa \sigma(M^0) \leq \sigma(M^0_g) \leq \mc{K} \sigma(M^0)
\end{align}
where $\sigma(M^0) = \sigma(\overline{W} \otimes I_m) \cup \sigma(\overline{A} \otimes I_m)$,
and, from Lemma~\ref{lem_sc}, matrix~$M^0$ has~$m$ eigenvalues associated with~$m$ dimensions of vector states~$\mb{x}_i$, i.e.\footnote{Note that, following Assumption~\ref{ass_W}, all the eigenvalues of the matrix $M^0$ are real. To generalize the proof for the case of mirror directed graphs (digraphs) and strongly connected balanced digraphs (see Corollary~\ref{cor_mirror}), in the proof we consider the general case of complex eigenvalues.},
$$\operatorname{Re}\{\lambda_{2n,j}\} \leq \ldots \leq \operatorname{Re}\{\lambda_{3,j}\} < \lambda_{2,j} = \lambda_{1,j} = 0,$$ 
where $j=\{1,\ldots,m\}$. 
We check the spectrum variation under \eqref{eq_spect_k} by adding the (small) perturbation~$\alpha M_1$. In simple words, we want to see whether the zero eigenvalues~$\lambda_{1,j}$ and~$\lambda_{2,j}$ and other LHP eigenvalues of $M_0$ move to the RHP or not (see Fig.~\ref{fig_perturb}).  Denote these perturbed eigenvalues by~$\lambda_{1,j}(\alpha,t)$ and~$\lambda_{2,j}(\alpha,t)$ which satisfy \eqref{eq_spect_k}.
Following the results in Section~\ref{sec_aux}, we find the right eigenvectors corresponding to~$\lambda_{1,j}$ and~$\lambda_{2,j}$ of $M^0$ as,
\begin{align} \nonumber
V = [V_1~V_2] =\frac{1}{\sqrt{n}} \left(\begin{array}{cc}
	\mb{1}_n& \mb{0}_n \\
	\mb{0}_n & \mb{1}_n 
\end{array} \right)\otimes I_m,
\end{align}
and the left eigenvectors as~$V^\top$ using Lemma~\ref{lem_laplacian}. Note that~$V^\top V=I_{2mn}$. Then, since $\Xi$ is a diagonal vector, the associated ones with $M_g^0$ are\footnote{Note that, since we only care about the sign of $\partial_\alpha \lambda_{1,j}$ and $\partial_\alpha \lambda_{2,j}$, the given eigenvectors are not necessarily unit vectors. Therefore, without loss of generality and to simplify the math derivation, we skipped the term 
that normalizes the eigenvectors. 
}
\begin{align} \nonumber
V_g = [V_{g1}~V_{g2}] =  \left(\begin{array}{cc}
	\mb{1}_n& \mb{0}_n \\
	\mb{0}_n &  \mb{1}_n 
\end{array} \right)\otimes I_m,
\end{align}
as the right eigenvectors and the same $V^\top$ as the left eigenvectors. This follows from Remark~\ref{rem_lambda_xi}.
From \eqref{eq_M_g},~$\partial_{\alpha} M_g(\alpha)|_{\alpha=0}=M_1$  with $M_1$ independent of the nonlinearity $g(\cdot)$. Then, from Lemma~\ref{lem_dM},

\small
\begin{eqnarray} \label{eq_dmalpha}
    		V^\top M_1 V_g= \left(\begin{array}{cc}
    \mb{0}_{m\times m}	& \mb{0}_{m\times m} \\
    \times 	& -(\mb{1}_n \otimes I_m)^\top H (\mb{1}_n \otimes I_m)
    \end{array} \right).
    \end{eqnarray} \normalsize
From Assumption~\ref{ass_f(x)} and \ref{ass_g_sector},
    \begin{equation} \label{eq_sum_df}
    	-(\mb{1}_n \otimes I_m)^\top H  ( \mb{1}_n \otimes I_m)= -\sum_{i=1}^n \nabla^2  f_i(\mb{x}_i) \prec 0,
    \end{equation}
    From Lemma~\ref{lem_dM}~$\partial_{\alpha} \lambda_{1,j}|_{\alpha=0}$ and~$\partial_{\alpha} \lambda_{2,j}|_{\alpha=0}$ depend on the eigenvalues of~$V^\top M_1 V_g$ in \eqref{eq_dmalpha}. This matrix is lower triangular with~$m$ zero eigenvalues and~$m$ negative eigenvalues (from~\eqref{eq_sum_df}). Therefore,~${\partial_{\alpha} \lambda_{1,j}|_{\alpha=0} = 0}$ and~${\partial_{\alpha} \lambda_{2,j}|_{\alpha=0}<0}$,
    i.e., the perturbation~$\alpha M_1$ moves the~$m$ zero eigenvalues~$\lambda_{2,j}(\alpha,t)$ of~$M_g$ toward the LHP while other $m$ zero eigenvalues $\lambda_{1,j}(\alpha,t)$ remain zero. This completes the proof.
\end{proof}

\begin{cor} \label{cor_eig}
    Note that Eq.~\eqref{eq_dmalpha}-\eqref{eq_sum_df} in the proof of Theorem~\ref{thm_zeroeig} are a function of the left and right eigenvectors of $\overline{W}_q$ and $\overline{A}_q$ (associated with the zero eigenvalue) and the $H$ matrix of the cost function. Given a PD $H$ matrix with $\nabla^2  f_i(\mb{x}_i) \succ 0$ (from Assumption~\ref{ass_f(x)}), any set of non-negative vectors $V^\top,V_g$ with $\sum_{i=1}^n v_i v_{g,i} > 0$ satisfies Eq.~\eqref{eq_sum_df}. Therefore, the stability results in Theorem~\ref{thm_zeroeig} hold for any Laplacian matrix whose (left-right) eigenvectors $V^\top,V_g$ satisfy $v_i v_{g,i} \geq 0,\forall i$ and $\sum_{i=1}^n v_i v_{g,i} > 0$. In case of homogeneous quadratic costs (i.e., $H = \gamma I$) this sufficient condition simplifies to only $\sum_{i=1}^n v_i v_{g,i} > 0$.  
\end{cor}
Some relevant discussions to Corollary~\ref{cor_eig} are given in \cite[Corollary~2]{SensNets:Olfati04}
on weighted-average consensus.
\begin{figure} [t]
		\centering
		\includegraphics[width=2.7in]{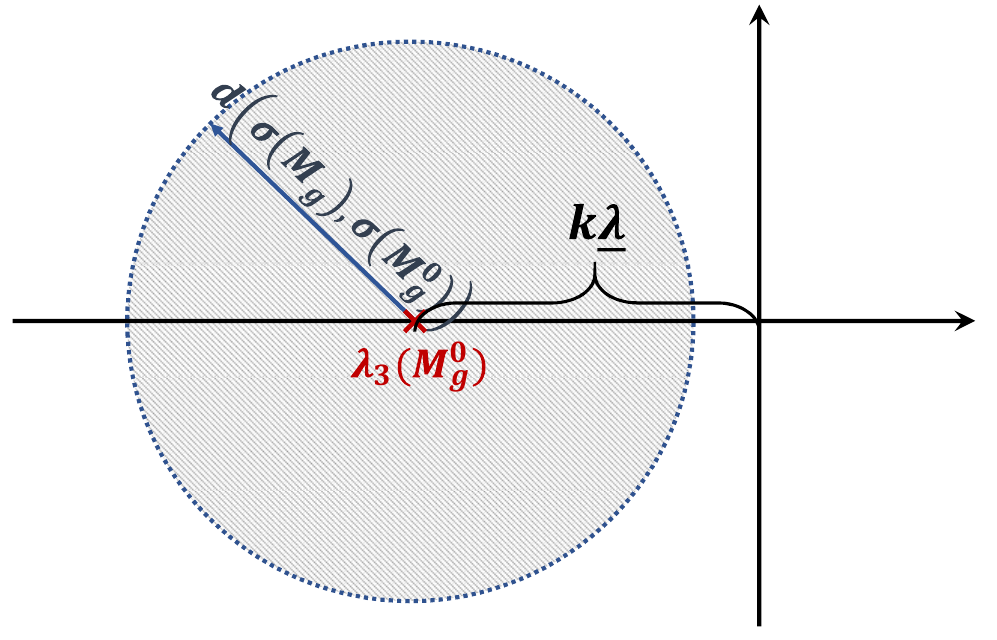}
		\caption{This figure illustrates the eigenvalue perturbation analysis in the proof of Theorem~\ref{thm_zeroeig} with $d(\sigma(M_g),\sigma(M_g^0))$ as the optimal matching distance and $\kappa \lambda_{3,j}(\alpha,t)$ as the maximum negative eigenvalue of $M_g^0$. For this example, we assume $\lambda_{3,j}(\alpha,t)$ is real-valued. }
		\label{fig_perturb}
		\vspace{-0.79cm}
\end{figure}
Next, the following finds the bound on $\alpha$ such that other eigenvalues of $M_g^0$ under perturbation $M_1$ remain on the LHP; this is better illustrated in Fig.~\ref{fig_perturb}.

In \cite{dsvm}, we found a bound on $\alpha$ for the linear case, which we recalculate here for the nonlinear case. Let~${\gamma = \max_{1\leq i\leq nm} \sum_{j=1}^{nm} |H_{ij}|}$ and ${\underline{\lambda}=\min_{1\leq j\leq m}|\operatorname{Re}\{\lambda_{3,j}\}|}$ with $\lambda_{3,j} \in \sigma(M^0)$. From \eqref{eq_beta_M} (and \eqref{eq_spect_k}), we have $\kappa \underline{\lambda} \leq |\operatorname{Re}\{\lambda_{3,j}(\alpha,t)\} |$, and for eigenvalues of $M_g^0$ to remain in the LHP we need $d(\sigma(M_g),\sigma(M^0_g)) < \kappa \underline{\lambda}$. Then, from Lemma~\ref{lem_dbound} and Fig.~\ref{fig_perturb}, the optimal matching distance $d(\sigma(M_g),\sigma(M^0_g)) < \kappa \underline{\lambda}$ (i.e., the real parts of the eigenvalues ${\operatorname{Re}\{\lambda_{3,j}(\alpha,t)\},\ldots,\operatorname{Re}\{\lambda_{2n,j}(\alpha,t)\}}$ of $M^0_g$ are negative) for~$0<\alpha<\overline{\alpha}_g$ with $\overline{\alpha}_g$ defined as
\begin{align} \nonumber
\overline{\alpha}_g &= \argmin_{\alpha>0} \big(| 4(2\mc{K}+2\gamma  \\
&+\max\{2\mc{K}+\gamma(2\mc{K}
+\alpha), 2\mc{K}+\alpha\}^{1-\frac{1}{nm}} \alpha^{\frac{1}{nm}}-\kappa \underline{\lambda}|\big), \label{eq_alphabar1}
\end{align}
for~$\gamma<1$ and,	
\begin{equation} \label{eq_alphabar2}
\overline{\alpha}_g = \argmin_{\alpha>0} |4(4\mc{K}+\gamma(4\mc{K}+\alpha))^{1-\frac{1}{nm}} (\alpha \gamma)^{\frac{1}{nm}}-\kappa \underline{\lambda}.|.
\end{equation}
for~$\gamma \geq 1$. These are defined based on $\lVert\cdot\rVert_\infty$ norm and Assumptions~\ref{ass_W} and~\ref{ass_f(x)} which gives $\lVert M_g^0 \rVert_\infty\leq 2\mc{K}(1+\gamma)$, $\lVert \alpha M^1 \rVert_\infty \leq \alpha \max\{ \gamma, 1\}$,
and thus, $\lVert M_g \rVert_\infty \leq \max\{2\mc{K}+\gamma(2\mc{K}+\alpha), 2\mc{K}+\alpha\}$.

Substituting spectral norm in Lemma~\ref{lem_dbound},  gives a more intuitive bound for the nonlinear case, which follows from 
\begin{align}
   4(\rVert M_0^g \lVert_2 + \rVert M^g \lVert_2)^{1-\frac{1}{nm}} \alpha^{\frac{1}{nm}} \rVert M_1\lVert_2^{\frac{1}{nm}} < \kappa \underline{\lambda},   
\end{align}
Assuming the admissible $\alpha$ such that the above holds, $\rVert M \lVert_2$ with respect to $\rVert M_0 \lVert_2$ is (at most) perturbed by $ \underline{\lambda}$, and from \eqref{eq_spect_k} we get
\begin{align} \nonumber
 4\mc{K}^{1-\frac{1}{nm}}(|\lambda_{2nm}| + |\lambda_{2nm}|+ \underline{\lambda})^{1-\frac{1}{nm}} \alpha^{\frac{1}{nm}}  \max\{1,\gamma\}^{\frac{1}{nm}} < \kappa \underline{\lambda}, 
\end{align}
which gives
\begin{align} \label{eq_alphag_spec_norm}
\overline{\alpha}_g \leq \frac{(\kappa \underline{\lambda})^{nm}}{4^{nm}\mc{K}^{nm-1}(2|\lambda_{2nm}| +\kappa \underline{\lambda})^{nm-1}   \max\{1,\gamma\}}
\end{align}
This gives an intuition on how the (sector-bound) ratio $\frac{\kappa}{\mc{K}}<1$ of the nonlinear mapping and the eigen-ratio $\frac{\underline{\lambda}}{|\lambda_{2nm}|}<1$  affect $\overline{\alpha}_g$ and the admissible range of ${\alpha}$. 
Under certain assumptions as in \cite{xi2017add}, one may find tighter bounds on the step-size $\alpha$. We first use the spectrum analysis in \cite[Appendix]{delay_est} to find the spectrum of $M_g(t,\alpha,q)$ in \eqref{eq_M_g}. By proper row/column permutations in \cite[Eq.~(18)]{delay_est}, $\sigma(M_g)$ follows from

\small \begin{align} \nonumber
\mbox{det}(&\alpha  I_{mn}) \mbox{det}(H(\overline{W}_{q,\Xi}\otimes I_m) +\\ &(\overline{A}_{q,\Xi} \otimes I_m - \alpha H -\lambda I_{mn}) (\frac{1}{\alpha})(\overline{W}_{q,\Xi} \otimes I_m -\lambda I_{mn})) = 0.
\end{align} \normalsize
We want to see whether any $\lambda \in \{\sigma(M^0_g)\backslash 0\}$ satisfies the above for an $\alpha$ value.
To simplify the analysis, we consider one-component state ($m=1$). Then,
\begin{align} \nonumber
    \det(\alpha H \overline{W}_{q,\Xi} &+ \overline{A}_{q,\Xi} \overline{W}_{q,\Xi}  - \alpha H \overline{W}_{q,\Xi}  - \lambda \overline{W}_{q,\Xi}   \nonumber\\
    &- \lambda \overline{A}_{q,\Xi} + \alpha \lambda H + \lambda^2 I_n) = 0
\end{align}
which simplifies to
\begin{align} \nonumber
 \det((\overline{W}_{q,\Xi}  -\lambda I_{n})(\overline{A}_{q,\Xi}  - \lambda I_{n}) +\alpha \lambda H ) = 0
\end{align}
Similar to the proof of Theorem~\ref{thm_zeroeig}, for stability we need to find the admissible range of  $\alpha$ values for which the eigenvalues $\sigma(M_g)$ remain in LHP, except one isolated zero eigenvalue.
We recall that the eigenvalues are continuous functions of the matrix elements~\cite{stewart_book}. Therefore, we find the roots of the above (say $\alpha=0$ and $\alpha = \overline{\alpha}_g$) for which we have more than one zero eigenvalue.  
Note that, for $\alpha = 0$, $\mbox{det}((\overline{W}_{q,\Xi}  - \lambda I_{n})(\overline{A}_{q,\Xi}  - \lambda I_{n})) = 0$ gives the eigen-spectrum as $\sigma(\overline{A}_{q,\Xi}) \cup \sigma(\overline{W}_{q,\Xi})$ with two zero eigenvalues. To find $\overline{\alpha}_g$, 
let assume the same adjacency matrix in \eqref{eq_xdot_g} and \eqref{eq_ydot_g} to simplify the problem (i.e., $w^q_{ij} = a^q_{ij}$). Then,
\begin{align} \nonumber
 \mbox{det}((\overline{W}_{q,\Xi}  - \lambda I_{n})^2 + \alpha \lambda H ) = 0
\end{align}
Note that the first term gives the eigenvalues of $\overline{W}_{q,\Xi}$ (in the LHP). Since $\alpha H$ is a diagonal matrix one can re-write the above (with some abuse of notation) as
\begin{align} 
 \mbox{det}((\overline{W}_{q,\Xi}  - \lambda I_{n} + \sqrt{\alpha |\lambda| H} )(\overline{W}_{q}  - \lambda I_{n} - \sqrt{\alpha |\lambda| H})) = 0
\end{align}
which simplifies to
\begin{align} 
 \mbox{det}(\overline{W}_{q,\Xi}  - \lambda(I_{n} \pm \sqrt{\frac{\alpha  H}{|\lambda|}} )) = 0
\end{align}
This simply means that the eigenvalues are perturbed towards the RHP by the factor of $\sqrt{\frac{\alpha  H}{|\lambda|}}$ (or adding the value $\sqrt{\alpha |\lambda| H}$). Therefore, the other root value follows from $\lambda I_{n} = \sqrt{\alpha |\lambda| H}$. 
In the linear case, one can claim that $\alpha_g = \argmin_{\alpha} |1 - \sqrt{\frac{\alpha  H}{|\lambda|}}| \geq \frac{\min \{|\lambda|\neq 0\}}{\max \{\partial^2 f_i\}} = \frac{|\lambda_2|}{\gamma}$ for $H \preceq \gamma I_{n}$ with $\lambda_2$ as the second largest eigenvalue of $M^0$. For the nonlinear case, recalling $\overline{W}_{q,\Xi} =  \overline{W}_{q} \Xi(t)$,
\begin{align} 
 \mbox{det}((\overline{W}_{q}  - \lambda I_{n} \Xi(t)^{-1} \pm \sqrt{\alpha |\lambda| H}\Xi(t)^{-1} )) = 0
\end{align}
which gives $\alpha_g \geq \min \{\frac{\kappa|\lambda_2|}{\gamma},\frac{|\lambda_2|}{\mc{K}\gamma}\}$.
This implies that, from perturbation analysis as in the proof of Theorem~\ref{thm_zeroeig}, one can claim for max bound on $\overline{\alpha}_g$ as
\begin{align} \label{eq_alphag_spec_norm1}
\overline{\alpha}_g  \leq \min \{\frac{\kappa|\lambda_2|}{\gamma},\frac{|\lambda_2|}{\mc{K}\gamma}\}
\end{align}
other $n-1$ eigenvalues of $M_g$ are in the LHP. This gives the admissible range to ensure stability\footnote{Note that both Eq.~\eqref{eq_alphag_spec_norm} and ~\eqref{eq_alphag_spec_norm1} give a bound on $\alpha_g$, but based on two different approaches. One can choose the one that gives an easier bound to satisfy the convergence. }. This is used in the next theorem.

\begin{thm} \label{thm_zeroeig2}
Let Assumptions~\ref{ass_f(x)}-\ref{ass_W} hold. For sufficiently small~$0 < \alpha_g <\overline{\alpha}_g$, the proposed nonlinear dynamics~\eqref{eq_xdot_g}-\eqref{eq_ydot_g} converges to~$[\mb{x}^*;\mb{0}_{nm}]$ with optimizer~$\mb{x}^*$ as the solution to~\eqref{eq_prob}. 
\end{thm}
\begin{proof}
    The proof follows from Theorems~\ref{thm_zeroeig}, circle criterion, and the admissible range of $\alpha$ for which all the eigenvalues of~$M_g$ have non-positive real-parts, $\forall t,q$, and the algebraic multiplicity of zero eigenvalue is~$m$. Note that for the circle criterion $A = M^1$, $C = M^0_g$, $\kappa \mb{v} \leq \phi(\mb{v},t) \leq \mc{K} \mb{v}$.
    Then the proof of stability follows from Lyapunov-type analysis similar to \cite{dsvm}; define the positive-definite Lyapunov function ${V(\boldsymbol{\delta}) = \frac{1}{2} \boldsymbol{\delta}^\top \boldsymbol{\delta} =  \frac{1}{2}\lVert\boldsymbol{\delta}\rVert_2^2}$ 
with residual~$\boldsymbol{\delta} = [\mb{x};\mb{y}]-[\mb{x}^*;\mb{0}_{mn}] \in \mathbb{R}^{2mn}$, i.e., the difference of the system state and the optimal state at every time.
Recall that~$[\mb{x}^*;\mb{0}_{nm}]$ is an invariant state of the proposed dynamics and 
$${\dot{\boldsymbol{\delta}} = [\dot{\mb{x}};\dot{\mb{y}}]-[\dot{\mb{x}}^*;\mb{0}_{mn}]} = {M_g ([\mb{x};\mb{y}]-[\mb{x}^*;\mb{0}_{mn}]) = M_g \boldsymbol{\delta}}$$
where $M_g [\mb{x}^*;\mb{0}_{mn}] = \mb{0}_{2mn}$.
Then, at every operating point $t>0$ the derivative satisfies
${\dot{V} = \boldsymbol{\delta}^\top \dot{\boldsymbol{\delta}}=  \boldsymbol{\delta}^\top M_g \boldsymbol{\delta}}$. Since only~${{\lambda}_{1,j}(\alpha,t)=0}$, while other eigenvalues are in LHP, i.e.,~${Re\{{\lambda}_{i,j}(\alpha,t)\}<0}$, for~${2\leq i\leq2n,1\leq j\leq m}$,
 one can show that,
\begin{eqnarray} \label{eq_Re2}
	\boldsymbol{\delta}^\top M_g \boldsymbol{\delta} \leq \max_{1\leq j\leq m}\operatorname{Re}\{{\lambda}_{2,j}(\alpha,t)\} \boldsymbol{\delta}^\top  \boldsymbol{\delta}. 
\end{eqnarray}
where $\lambda_{2,j}(\alpha,t)$ denotes the largest nonzero eigenvalue of $M_g$. Note that although $M_g$ varies in time, we proved in Theorem~\ref{thm_zeroeig} that $\operatorname{Re}\{{\lambda}_{2,j}(\alpha,t)<0\}$ for ${1\leq j\leq m}$, implying that~${\dot{V}< 0}$ for~${\delta \neq \mb{0}_{2mn}}$.
Then, similar to LaSalle’s invariance principle for hybrid dynamic systems \cite{goebel2009hybrid}, we have 
\begin{align} \label{eq_vdot}
    {\dot{V} =  0 \Leftrightarrow \boldsymbol{\delta} = \mb{0}_{2mn}}
\end{align}
and the dynamics converges to its invariant set~${\{\boldsymbol{\delta} = \mb{0}_{2mn}\}}$. This proves the theorem\footnote{One can refer to the contraction theory
to similarly prove the results as  $\overline{W}_q$, $\overline{A}_q$ are symmetric matrices (from Assumption~\ref{ass_W}) and one can show that the system trajectories converge towards each other \cite{kalman_conj}.}. 
\end{proof}

\begin{cor} \label{cor_mirror}
As a follow-up to Corollary~\ref{cor_eig}, it is possible to extend Theorem~\ref{thm_zeroeig2} to strongly-connected weight-balanced (WB) directed graphs by using~\cite[Theorem~8]{SensNets:Olfati04} (see Courant-Fischer Theorem on balanced mirror digraphs \cite{SensNets:Olfati04}). This is one of our ongoing research directions and is only shown via simulation in Section~\ref{sec_sim}.
\end{cor}

Recall that, in the above proof, the Lyapunov function $V$ remains continuous at the jump (switching) points under the proposed hybrid dynamics.
The nonlinear mapping $g(\cdot)$ in \eqref{eq_xdot_g} and \eqref{eq_ydot_g} can be different in general. In the case of 
different sector bound parameters $\kappa_1,\mc{K}_1$ and  $\kappa_2,\mc{K}_2$ the bound $\alpha_g$ can be found via similar calculations by considering $\min \{\kappa_1,\kappa_2\}$ and $\max \{\mc{K}_1,\mc{K}_2\}$ instead.

\begin{rem}
    The weight-symmetric of networks (and WB condition in general) is a milder condition than bi-stochasticity \cite{icrom22_drop}. This is because, in case of link removal and packet drops, e.g., in unreliable networks, the weight-symmetric condition may still hold but bi-stochasticity is not necessarily satisfied. In other words, bi-stochastic weights in volatile networks need to be redesigned via specific algorithms to ensure that the weights of incoming and outgoing links over the new network still sum to one, e.g., via the algorithms proposed in \cite{6426252,cons_drop_siam}. This paper on the other hand needs only symmetric weights and still remains weight-symmetric in case of removing any bidirectional link. This is an improvement over the existing literature.
\end{rem}

\section{Simulations} \label{sec_sim}
For the simulation, we recall the same setup in \cite{dsvm}  based on the example given in~\cite{russell2010artificial}. Consider~${N=200}$ uniformly distributed sample data points ${\boldsymbol{\chi}_i=[{\chi}_i(1);{\chi}_i(2)]}$ as shown in Fig.~\ref{fig_data}(Left), classified as blue *'s and  red~o's. In~$\mathbb{R}^2$, no linear SVM classifier can be defined, and applying nonlinear mapping~$\phi(\boldsymbol{\chi}_i) = [{\chi}_i(1)^2;{\chi}_i(2)^2;\sqrt{2}{\chi}_i(1){\chi}_i(2)]$ (with kernel
$K(\boldsymbol{\chi}_i,\boldsymbol{\chi}_j)=(\phi(\boldsymbol{\chi}_i)^\top \phi(\boldsymbol{\chi}_j))^2$) the data points can be linearly classified in~$\mathbb{R}^3$ ($m=4$), see Fig.~\ref{fig_data}(Right). 
\begin{figure}[tbp]
	\centering
	\includegraphics[width=2in]{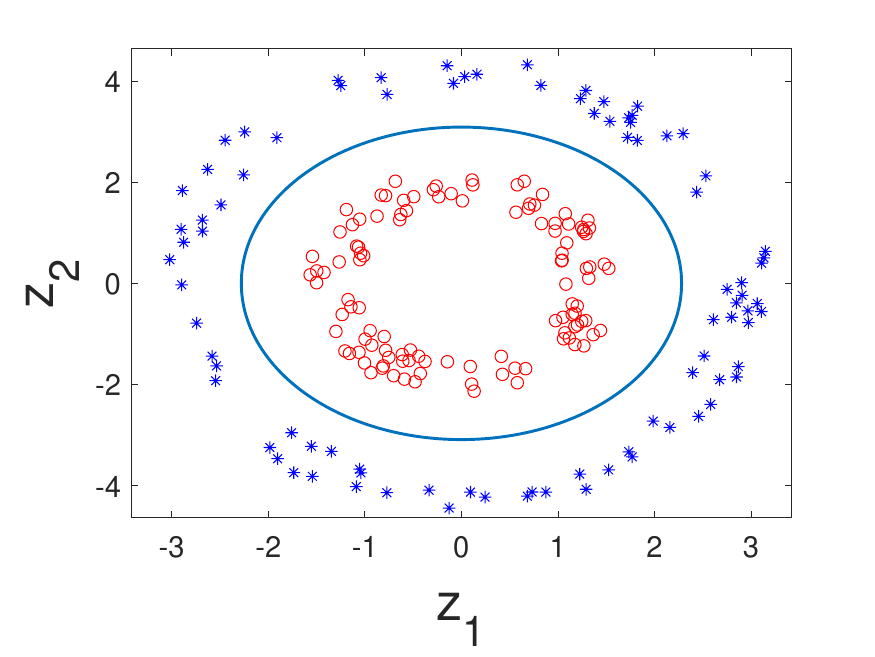}
	\includegraphics[width=2in]{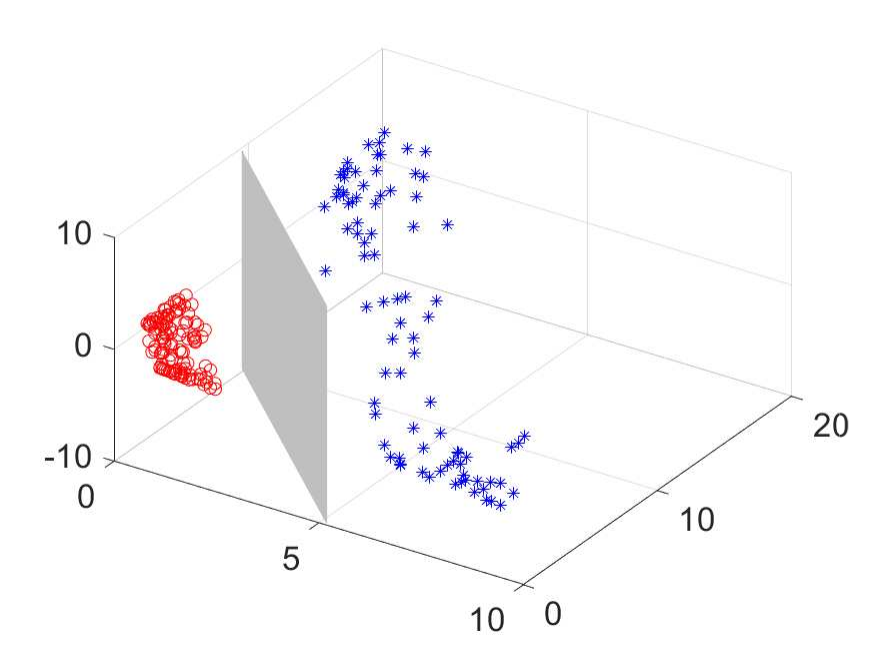}	
	\caption{(Left) This figure shows the training data and the optimal nonlinear classifier (Left) in 2D as the ellipse and  (Right) (mapped) in 3D space as a linear hyperplane. Our nonlinear D-SVM under logarithmic quantization optimally classifies the data via the grey hyperplane. This hyperplane, inversely mapped back into 2D, represents the ellipse in the left figure.}  \label{fig_data}
\end{figure}
For the purpose of the simulation, we assume nonlinear log quantization at the links in the proposed dynamics~\eqref{eq_xdot_g}-\eqref{eq_ydot_g}, i.e.,  
$$g(z)= q_l(z) = \mbox{sgn}(z)\exp\left(\rho\left[\dfrac{\log(|z|)}{\rho}\right] \right)$$
 with $[\cdot]$ as rounding to the nearest integer, $\mbox{sgn}(\cdot)$  as the sign function, and $\rho$ as the quantization level. This nonlinearity satisfies Assumption~\ref{ass_g_sector}.
The network $\mc{G}_q$ of~${n=5}$ agents is considered as a dynamic $2$-hop WB digraph such that, using MATLAB's \texttt{randperm}, the nodes' permutation (and neighbourhood) changes every $0.001$ sec.
Every agent finds the optimal hyperplane parameters~$\mb{x}_i = [\boldsymbol{\omega}_i^\top;\nu_i]$ ($\boldsymbol{\omega}_i \in \mathbb{R}^3$) to minimize its local loss function defined by~\eqref{eq_svm_dist} and, then, shares~$\mb{x}_i$ and $\mb{y}_i$ over $\mc{G}_q$.  The simulation parameters are as follows: ${\alpha=6}$,  $\mu=2$, $C=1$, $\rho=1$ with initialization $\mb{y}(0) = \mb{0}_{20}$ and $\mb{x}(0) = \texttt{rand(20,1)}$.
The time-evolution of~$\mb{x}_i  = [\boldsymbol{\omega}_i^\top;\nu_i] \in \mathbb{R}^4$, loss function~$F(\mb{x})$, and sum of the gradients~$\sum_{i=1}^{5} \nabla f_i(\mb{x}_i) \in \mathbb{R}^4$ are given in  Fig.~\ref{fig_nonlin}. 

\begin{figure} 
	\centering
 	\includegraphics[width=3in]{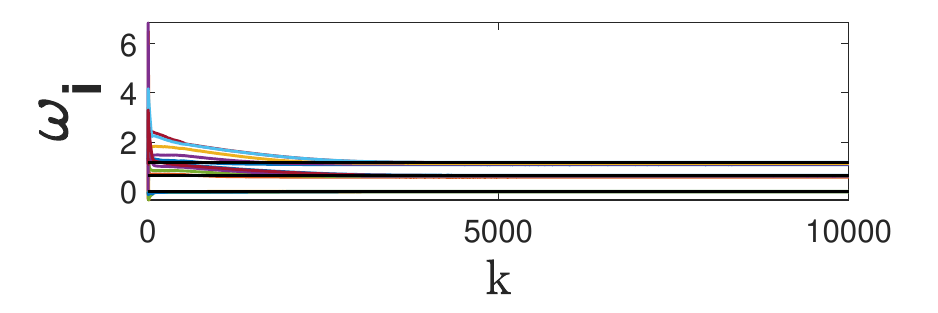} 
 	\includegraphics[width=3in]{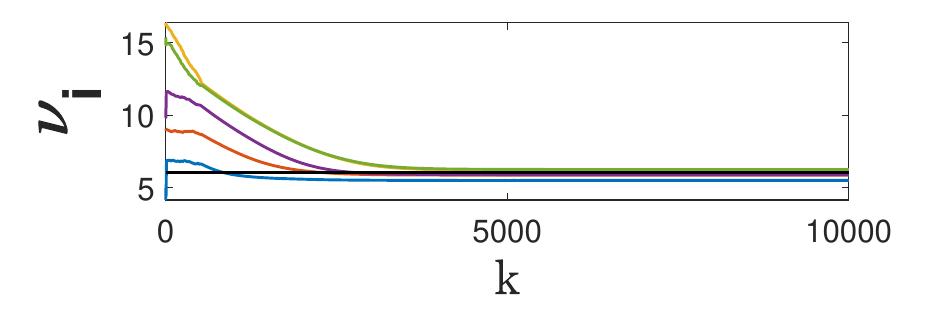}	
	\includegraphics[width=3in]{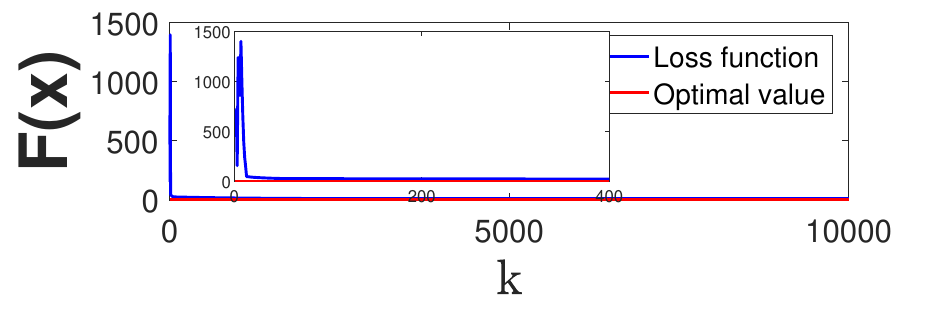}	
	\includegraphics[width=3in]{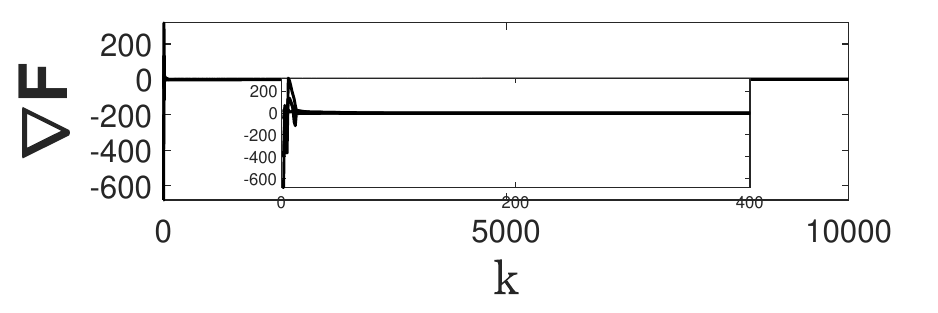}
	\caption{The time-evolution of the parameters ~${\omega}_i$ and~$\nu_i$ under the nonlinear D-SVM classifier dynamics~\eqref{eq_xdot_g}-\eqref{eq_ydot_g} subject to log quantization. The overall loss function~$F(\mb{x})$, the optimal values based on the centralized SVM, and the sum of the local gradients~$\sum_{i=1}^{5} \nabla f_i(\mb{x}_i)$ are also shown. }  \label{fig_nonlin} \vspace{-0.5cm}
\end{figure}
The agents reach consensus on the optimal value ${\overline{\mb{x}}^*=[\overline{\omega}(1),\overline{\omega}(2),\overline{\omega}(3),\overline{\nu}]^\top}$,
which represents the separating ellipse  ${\overline{\omega}(1)z_1^2 + \overline{\omega}(2)z_2^2 -\overline{\nu}=0}$ ($z_1$ and~$z_2$ as the Cartesian coordinates in~$\mathbb{R}^2$).
For comparison, D-SVM under linear setup \cite{dsvm} is also given in Fig.~\ref{fig_lin}. It is clear that the nonlinearity only affects the agents' dynamics and not the overall Laplacian-gradient tracking performance given by dynamics \eqref{eq_sumydot}-\eqref{eq_sumxdot2}.  
\begin{figure} 
	\centering
 	\includegraphics[width=3in]{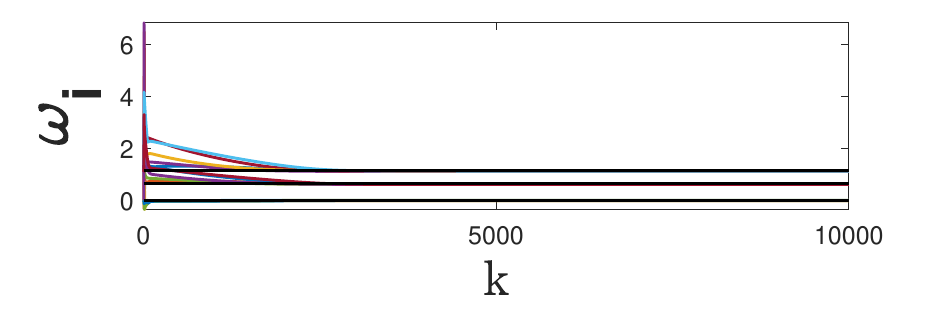} 
 	\includegraphics[width=3in]{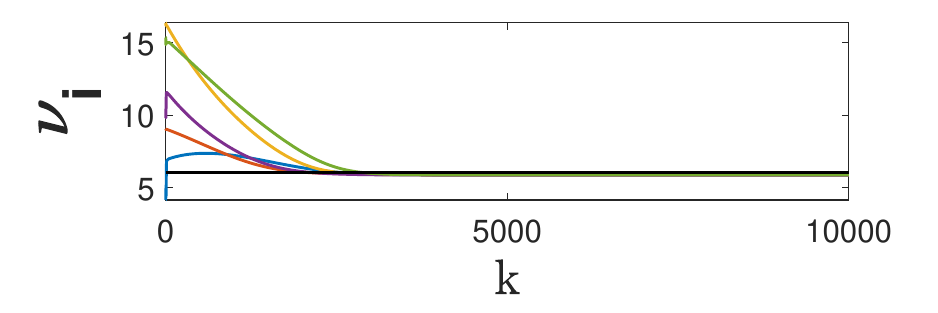}	
 	\includegraphics[width=3in]{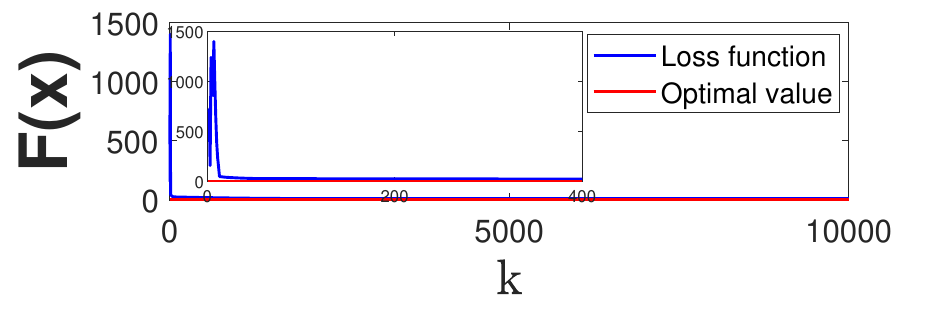}	
	\includegraphics[width=3in]{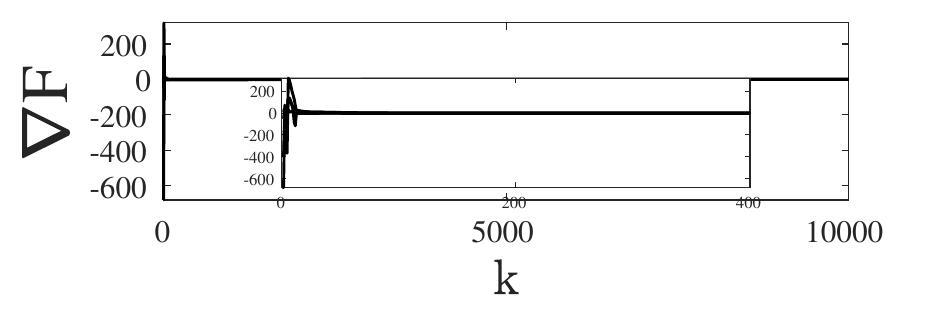}
	\caption{The time-evolution of the linear D-SVM classifier under dynamics~\eqref{eq_xdot}-\eqref{eq_ydot} (under the same parameters as in Fig.~\ref{fig_nonlin}) for comparison.}  \label{fig_lin} \vspace{-0.2cm}
\end{figure}
We used MATLAB R2021b for the simulation and \texttt{CVX} toolbox for centralized SVM (for comparison). As clear from the figures, the distributed solutions converge to the centralized optimal solution in both linear and quantized cases for admissible $\alpha,\rho$ values.
Although, since the loss function \eqref{eq_svm_dist} grows exponentially, this results in some inaccuracy in MATLAB calculations and minor residual in the figures.

Next, we perform some sensitivity analysis on $\alpha_g$  by adding an extra link to reduce the network diameter to change $\sigma(M_0^g)$ and, also, by changing the sector-bound-ratios by tuning the level $\rho$ in log quantization since  $\kappa = 1-\dfrac{\rho}{2} \leq \dfrac{q_l(z)}{z}\leq 1+\dfrac{\rho}{2} = \mc{K}$ (Assumption~\ref{ass_g_sector}). We also tuned the network connectivity by changing the number of neighbours by linking the $2$-hop and $3$-hop neighbouring nodes.  The results are shown in Fig.~\ref{fig_sensit} for $\alpha = 0.2$ and discrete time-step $\eta = 0.05$. Note that, from the figure, large values of $\frac{\mc{K}}{\kappa}$ (Table~\ref{tab_comp_K}) and $\frac{|\lambda_{2mn}|}{\underline{\lambda}}$ (Table~\ref{tab_comp_eig}) reduces the admissible range of $\alpha$ and makes the solution unstable for the given $\alpha$.  
\begin{figure} 
	\centering
	\includegraphics[width=2.75in]{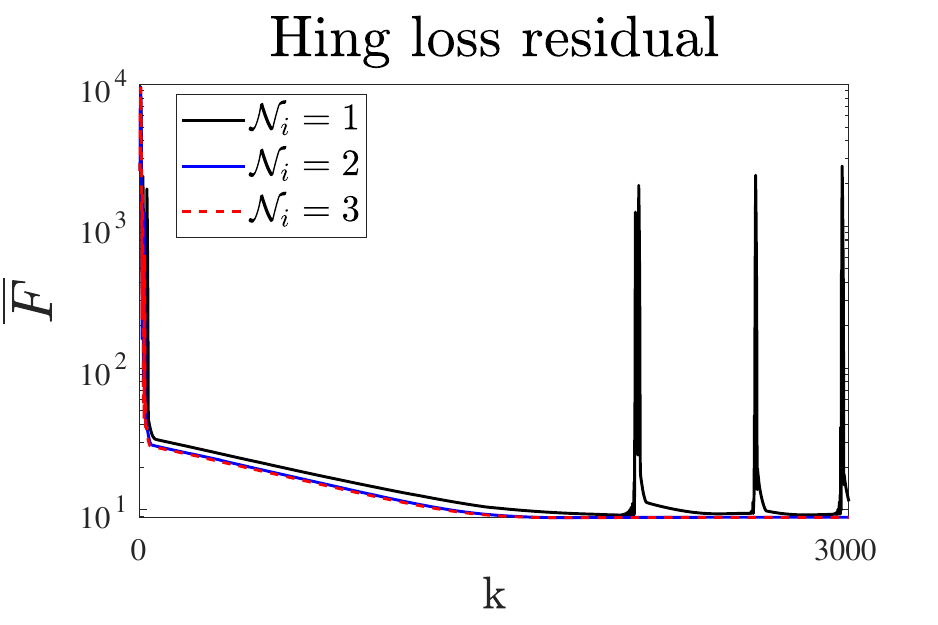}
	\includegraphics[width=2.75in]{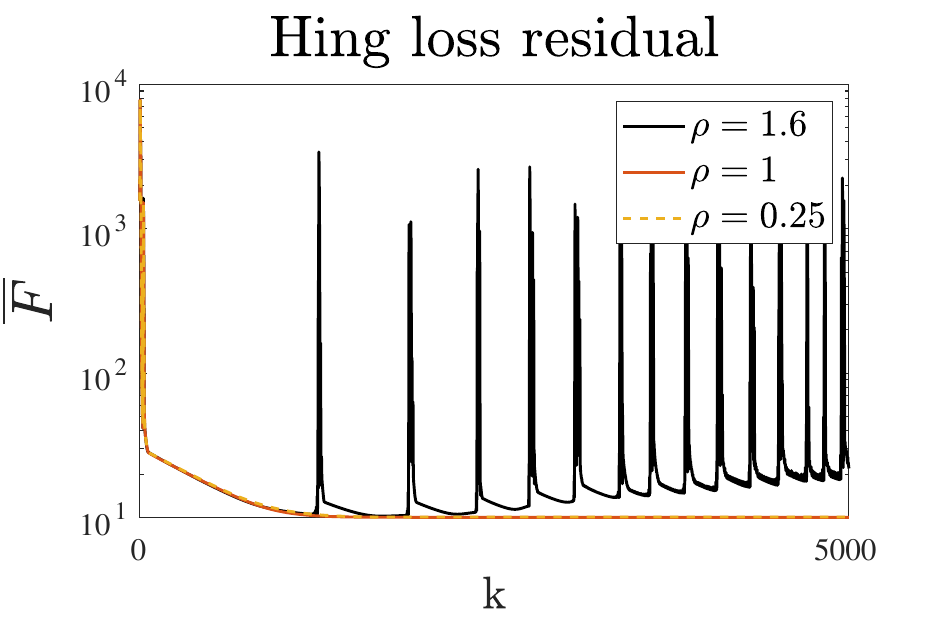} 
	\caption{This figure shows how the change in network parameter $\mc{N}_i$ (and its eigen-spectrum) and the sector-bounds (the quantization parameter $\rho$) on the nonlinear mapping affect the stability of the dynamics~\eqref{eq_xdot_g}-\eqref{eq_ydot_g}. }  \label{fig_sensit} \vspace{-0.5cm}
\end{figure}
\begin{table} [h] 
		\centering
		\caption{Change in the sector-bound-ratio by changing the log quantization level $\rho$}
		\label{tab_comp_K}
		\begin{tabular}{|c|c|c|c|c|}
			\hline
		   Sector-bound-ratio & $\rho=1.6$  & $\rho=1$ & $\rho=0.25$ \\
			\hline
			$\frac{\mc{K}}{\kappa}$ & 9 & 3 &  1.28  \\
			\hline
			\hline
		\end{tabular}
\end{table}
\begin{table} [h] 
		\centering
		\caption{Change in the eigen-ratio by changing the network connectivity (number of $\mc{N}_i$-hop neighbours) }
		\label{tab_comp_eig}
		\begin{tabular}{|c|c|c|c|c|}
			\hline
		   Eigen-ratio & $\mc{N}_i=1$  & $\mc{N}_i=2$ & $\mc{N}_i=3$ \\
			\hline
			$\frac{\mbox{Max}\{|\lambda_{2mn}(t)|\}}{\underline{\lambda}}$ & 13733 & 2971 &  1413  \\
			\hline
			$\frac{\mbox{Mean}\{|\lambda_{2mn}(t)|\}}{\underline{\lambda}}$ & 328 & 65 &  33  \\
			\hline
			\hline
		\end{tabular}
\end{table}


\section{Discussions and Future Directions} \label{sec_con}
\subsection{Concluding Remarks}
The consensus optimization algorithm in this work is shown to converge under strongly sign-preserving odd nonlinearities.  
Note that ``non-strongly'' sign-preserving nonlinearities may cause steady-state residuals as discussed in some recent works. For example, for general non-quadratic optimization subject to uniform quantization the existing methods only guarantee convergence to the $\varepsilon$-neighborhood of the exact optimizer $\mb{x}^*$, i.e., the algorithm may result in steady-state residual bounded by $\varepsilon$ \cite{magnusson2018communication}. As an example consider uniform quantization defined as $g(z)=\rho\left[\dfrac{z}{\rho}\right]$. This nonlinear function although is sign-preserving, is not ``strongly'' sign-preserving as $\frac{g(z)}{z}=0$ for $-\frac{\rho}{2} < z < \frac{\rho}{2}$. This implies that in the proof of Theorem~\ref{thm_zeroeig2} the Eq. \eqref{eq_vdot} does not necessarily hold (invariant set includes an $\varepsilon$-neighborhood of $\delta = \mb{0}$) and exact convergence is not guaranteed. Finding this $\varepsilon$ bound on the optimizer as a function of quantization level $\rho$ and loss function parameters is an interesting direction for future research. Recall that the logarithmic quantizer, on the other hand, is ``strongly'' sign-preserving and sector-bounded (satisfying Assumption~\ref{ass_g_sector}). Thus, one can prove its convergence under \eqref{eq_xdot_g}-\eqref{eq_ydot_g} as shown in our simulations. Note that, although locally non-Lipschitz, log quantization satisfies Assumption~\ref{ass_g_sector} with no chattering effect (in contrast to \cite{rahili_ren,ning2017distributed,taes2020finite}). The results can be further extended to optimize the loss function (least mean-square-error problem) in distributed estimation \cite{doostmohammadian2021distributed,my_acc} and target tracking \cite{doostmohammadian2022linear}.

\subsection{Future Research}
Generalizing Assumption~\ref{ass_g_sector} for less constrained link nonlinearity condition,  extending Corollary~\ref{cor_mirror} (and Assumption~\ref{ass_W}), and convergence under data-transmission delays over the links are our future research directions.

\section*{Statements \& Declarations}

The authors would like to thank Wei Jiang, Evagoras Makridis,  Muwahida Liaquat, Themistoklis Charalambous, and Usman Khan for their comments.

The authors declare that no funds, grants, or other support were received during the preparation of this manuscript.

The authors have no relevant financial or non-financial (or other competing) interests to disclose.

All authors contributed to the study conception, preparation, and writing. 
All authors read and approved the final manuscript.

This research involves NO human or animal subjects. 

Data sharing not applicable to this article as no datasets were generated or analysed during the current study.

\bibliographystyle{spmpsci} 
\bibliography{bibliography}

\end{document}